\documentclass[10pt]{article}
\usepackage{amsmath,amssymb}
\usepackage{fullpage}
\usepackage{natbib}
\usepackage{graphicx}

\def\R{{\mathbb{R}}}
\def\pr{{\rm Pr}}
\def\E{{\mathbb E}}
\def\X{{\mathcal X}}

\newtheorem{thm}{Theorem}
\newtheorem{lemma}[thm]{Lemma}
\newtheorem{cor}[thm]{Corollary}

\newenvironment{proof}{\noindent {\sc Proof:}}{$\Box$ \medskip}

\title{Randomized partition trees for exact nearest neighbor search}
\author{Sanjoy Dasgupta and Kaushik Sinha \\ \texttt{\{dasgupta,ksinha\}@cs.ucsd.edu} \\ UC San Diego}

\begin{document}

\maketitle

\begin{abstract}
The $k$-d tree was one of the first spatial data structures proposed for nearest neighbor search. Its efficacy is diminished in high-dimensional spaces, but several variants, with randomization and overlapping cells, have proved to be successful in practice. We analyze three such schemes. We show that the probability that they fail to find the nearest neighbor, for any data set and any query point, is directly related to a simple potential function that captures the difficulty of the point configuration. We then bound this potential function in two situations of interest: the first, when data come from a doubling measure, and the second, when the data are documents from a topic model.
\end{abstract}

\section{Introduction}

The problem of nearest neighbor search has engendered a vast body of algorithmic work. In the most basic formulation, there is a set $S$ of $n$ points, typically in an Euclidean space $\R^d$, and any subsequent query point must be answered by its nearest neighbor (NN) in $S$. A simple solution is to store $S$ as a list, and to address queries using a linear-time scan of the list. The challenge is to achieve a substantially smaller query time than this.

We will consider a prototypical modern application in which the number of points $n$ and the dimension $d$ are both large. The primary resource constraints are the size of the data structure used to store $S$ and the amount of time taken to answer queries. For practical purposes, the former must be $O(n)$, or maybe a little more, and the latter must be $o(n)$. Secondary constraints include the time to build the data structure and, sometimes, the time to add new points to $S$ or to remove existing points from $S$.

A major finding of the past two decades has been that these resource bounds can be met if it is enough to merely return a {\it $c$-approximate nearest neighbor}, whose distance from the query is at most $c$ times that of the true nearest neighbor. One such method that has been successful in practice is {\it locality sensitive hashing} (LSH), which has space requirement $n^{1+\rho}$ and query time $O(n^\rho)$, for $\rho \approx 1/c^2$~\citep{AI08}. Another such method is the {\it balanced box decomposition tree}, which takes $O(n)$ space and answers queries with an approximation factor $c = 1+\epsilon$ in $O((6/\epsilon)^d \log n)$ time~\citep{AMNSW98}.

In the latter result, an exponential dependence on dimension is evident, and indeed this is a familiar blot on the nearest neighbor landscape. One way to mitigate the curse of dimensionality is to consider situations in which data have low {\it intrinsic dimension} $d_o$, even if they happen to lie in $\R^d$ for $d \gg d_o$ or in a general metric space. A common assumption is that the data are drawn from a {\it doubling measure} of dimension $d_o$ (or equivalently, have {\it expansion rate} $2^{d_o}$); this is defined in Section~\ref{sec:doubling} below. Under this condition, \citet{KR02} have a scheme that gives exact answers to nearest neighbor queries in time $O(2^{3 d_o} \log n)$, using a data structure of size $O(2^{3 d_o} n)$. The more recent {\it cover tree} algorithm~\citep{BKL06}, which has been used quite widely, creates a data structure in space $O(n)$ and answers queries in time $O(2^{d_o} \log n)$. There is also work that combines intrinsic dimension and approximate search. The {\it navigating net}~\citep{KL04}, given data from a metric space of {\it doubling dimension} $d_o$, has size $O(2^{O(d_o)} n)$ and gives a $(1+\epsilon)$-approximate answer to queries in time $O(2^{O(d_o)} \log n + (1/\epsilon)^{O(d_o)})$; the crucial advantage here is that doubling dimension is a more general and robust notion than doubling measure.

Despite these and many other results, there are two significant deficiencies in the nearest neighbor literature that have motivated the present paper. First, existing analyses have succeeded at identifying, for a given data structure, highly specific families of data for which efficient exact NN search is possible---for instance, data from doubling measures---but have failed to provide a more general characterization. Second, there remains a class of nearest neighbor data structures that are popular and successful in practice, but that have not been analyzed thoroughly. These structures combine classical $k$-d tree partitioning with randomization and overlapping cells, and are the subject of this paper.

\subsection{Three randomized tree structures for exact NN search}

The $k$-d tree is a partition of $\R^d$ into hyper-rectangular cells, based on a set of data points~\citep{B75}. The root of the tree is a single cell corresponding to the entire space. A coordinate direction is chosen, and the cell is split at the median of the data along this direction (Figure~\ref{fig:kd-vs-rp}, left). The process is then recursed on the two newly created cells, and continues until all leaf cells contain at most some predetermined number $n_o$ of points. When there are $n$ data points, the depth of the tree is at most about $\log (n/n_o)$.

\begin{figure}
\begin{center}
\includegraphics[width=5in]{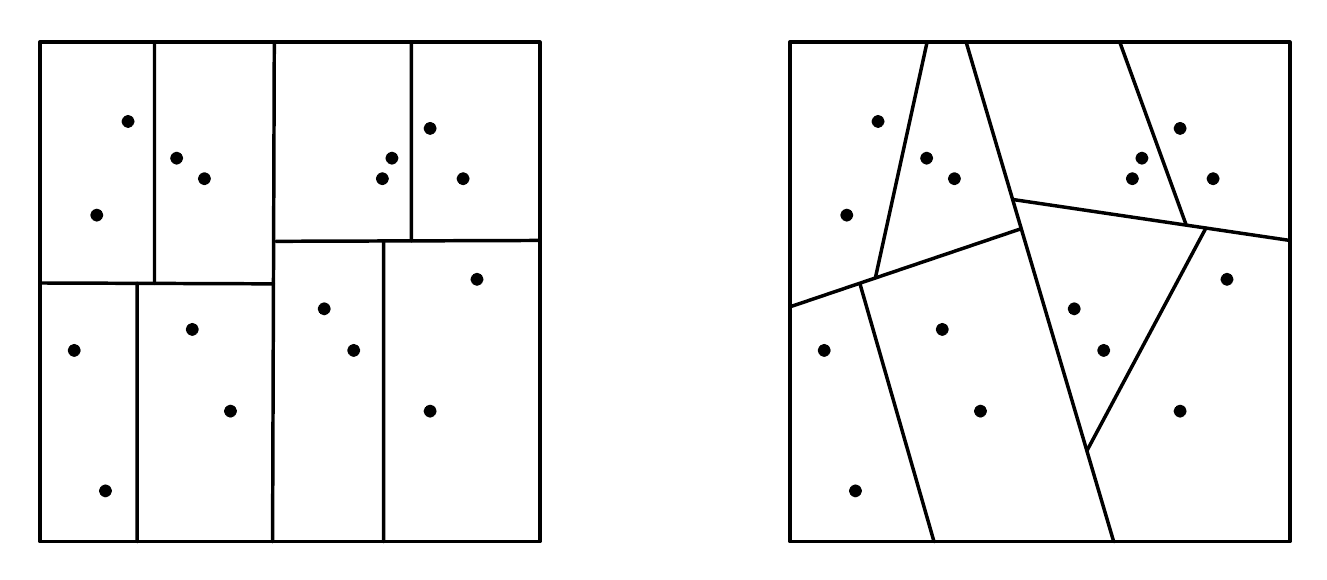}
\end{center}
\caption{Left: A $k$-d tree, with axis-parallel splits. Right: A variant in which the split directions are chosen randomly from the unit sphere.}
\label{fig:kd-vs-rp}
\end{figure}

Given a $k$-d tree built from data points $S$, there are several ways to answer a nearest neighbor query $q$. The quickest and dirtiest of these is to move $q$ down the tree to its appropriate leaf cell, and then return the nearest neighbor in that cell. This {\it defeatist search} takes time just $O(n_o + \log (n/n_o))$, which is $O(\log n)$ for constant $n_o$. The problem is that $q$'s nearest neighbor may well lie in a different cell, for instance when the data happen to be concentrated near cell boundaries. Consequently, the failure probability of this scheme can be unacceptably high.

\begin{figure}
\begin{quote}
\begin{tt}
\begin{tabbing}
Function MakeRPTree($S$) \\
\ \ \= If $|S| \leq n_o$: return leaf containing $S$ \\
    \> Pick $U$ uniformly at random from the unit sphere \\
    \> Pick $\beta$ uniformly at random from $[1/4,3/4]$ \\
    \> Let $v$ be the $\beta$-fractile point on the projection of $S$ onto $U$ \\
    \> Rule($x$) = ({\sc left} if $x \cdot U < v$, otherwise {\sc right}) \\
    \> $\mbox{LeftSubtree} = \mbox{MakeRPTree}(\{x \in S: \mbox{rule}(x) = \mbox{\sc left} \})$ \\
    \> $\mbox{RightSubtree} = \mbox{MakeRPTree}(\{x \in S: \mbox{rule}(x) = \mbox{\sc right} \})$ \\
    \> Return (Rule($\cdot$), LeftSubtree, RightSubtree)
\end{tabbing}
\end{tt}
\end{quote}
\caption{The random projection tree (RP tree)}
\label{alg:rptree}
\end{figure}

Over the years, some simple tricks have emerged, from various sources, for reducing the failure probability. These are nicely laid out by \citet{LMGY04}, who show experimentally that the resulting algorithms are effective in practice.

The first trick is to introduce randomness into the tree. Drawing inspiration from locality-sensitive hashing, \citet{LMGY04} suggest preprocessing the data set $S$ by randomly rotating it, and then applying a $k$-d tree (or related tree structure). This is rather like splitting cells along random directions as opposed to coordinate axes (Figure~\ref{fig:kd-vs-rp}, right). In this paper, we consider a data structure that uses random split directions as well as a second type of randomization: instead of putting the split point exactly at the median, it is placed at a fractile chosen uniformly at random from the range $[1/4, 3/4]$. The resulting structure (Figure~\ref{alg:rptree}) is almost exactly the {\it random projection tree} (or RP tree) of \citet{DF08}. That earlier work showed that in RP trees, the diameters of the cells decrease (down the tree) at a rate depending only on the intrinsic dimension of the data. It is a curious result, but is not helpful in analyzing nearest neighbor search, and in this paper we develop a different line of reasoning. Indeed, there is no point of contact between that earlier analysis and the one we embark upon here.

A second trick suggested by \citet{LMGY04} for reducing failure probability is to allow overlap between cells. This was also proposed in earlier work of \citet{MM01}. Once again, each cell $C$ is split along a direction $U(C)$ chosen at random from the unit sphere. But now, three split points are noted: the median $m(C)$ of the data along direction $U$, the $(1/2)-\alpha$ fractile value $l(C)$, and the $(1/2)+\alpha$ fractile value $r(C)$. Here $\alpha$ is a small constant, like $0.05$ or $0.1$. The idea is to simultaneously entertain a {\it median split}
$$ \mbox{left} = \{x: x \cdot U < m(C)\} \mbox{\ \ \ right} = \{x: x \cdot U \geq m(C)\} $$
and an {\it overlapping split} (with the middle $2 \alpha$ fraction of the data falling on both sides)
$$ \mbox{left} = \{x: x \cdot U < r(C)\} \mbox{\ \ \ right} = \{x: x \cdot U \geq l(C)\}.$$
In the {\it spill tree}~\citep{LMGY04}, each data point in $S$ is stored in multiple leaves, by following the overlapping splits. A query is then answered defeatist-style, by routing it to a single leaf using median splits.

Both the RP tree and the spill tree have query times of $O(n_o + \log (n/n_o))$, but the latter can be expected to have a lower failure probability, and we will see this in the bounds we obtain. On the other hand, the RP tree requires just linear space, while the size of the spill tree is $O(n^{1/(1-\lg (1+2\alpha))})$. When $\alpha = 0.05$, for instance, the size is $O(n^{1.159})$.

In view of these tradeoffs, we consider a further variant, which we call the {\it virtual spill tree}. It stores each data point in a single leaf, following median splits, and hence has linear size. However, each query is routed to multiple leaves, using overlapping splits, and the return value is its nearest neighbor in the union of these leaves.

\begin{figure}
\begin{center}
\includegraphics[width=6in]{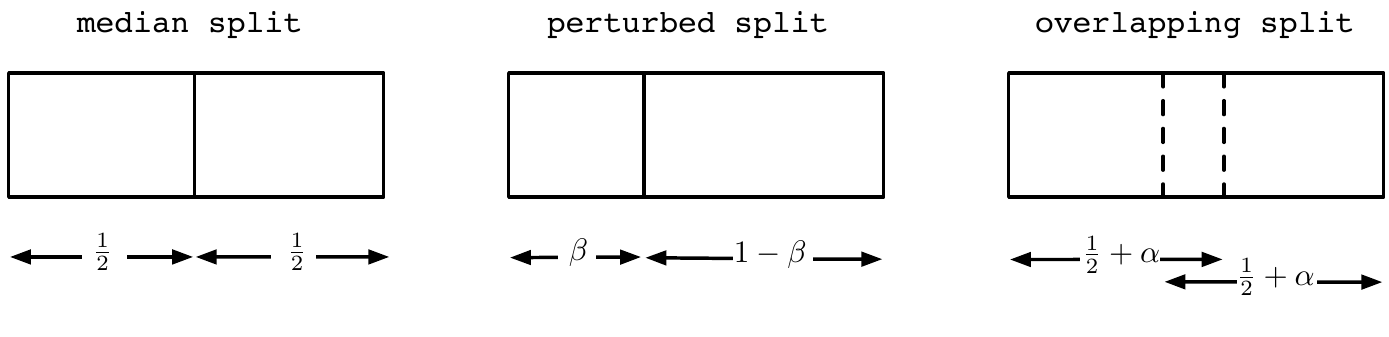}
\end{center}
\caption{Three types of split. The fractions refer to probability mass. $\alpha$ is some constant, while $\beta$ is chosen uniformly at random from $[1/4,3/4]$.}
\label{fig:splits}
\end{figure}

The various splits are summarized in Figure~\ref{fig:splits}, and the three trees use them as follows: 

\begin{center}
\begin{tabular}{c|cc}
                    & Routing data      & Routing queries \\ \hline
RP tree             & Perturbed split   & Perturbed split \\
Spill tree          & Overlapping split & Median split \\
Virtual spill tree  & Median split      & Overlapping split
\end{tabular}
\end{center}

One small technicality: if, for instance, there are duplicates among the data points, it might not be possible to achieve a median split, or a split at a desired fractile. We will ignore these discretization problems.

\subsection{Analysis of failure probability}

Our three schemes for nearest neighbor search---the RP tree and the two spill trees---can be analyzed in a simple and unified framework. Pick any data set $x_1, \ldots, x_n \in \R^d$ and any query $q \in \R^d$. The probability of failure, of not finding the nearest neighbor, can be shown to be directly related to the quantity
$$ \Phi(q, \{x_1, \ldots, x_n\}) \ = \ \frac{1}{n} \sum_{i=2}^n \frac{\|q-x_{(1)}\|}{\|q-x_{(i)}\|} ,$$
where $x_{(1)}, x_{(2)}, \ldots$ denotes an ordering of the $x_i$ by increasing distance from $q$. For RP trees, the failure probability is proportional to $\Phi \log (1/\Phi)$ (Theorem~\ref{thm:rp}); for the two spill trees, it is proportional to $\Phi$ (Theorem~\ref{thm:spill}). The results extend easily to the problem of searching for the $k$ nearest neighbors. Moreover, these bounds are roughly tight: a failure probability proportional to $\Phi$ is inevitable unless there is a significant amount of collinearity within the data (Corollary~\ref{cor:three-points}).

Let's take a closer look at this potential function. If $\Phi$ is close to 1, then all the points are roughly the same distance from $q$, and so we can expect that the NN query is not easy to answer. On the other hand, if $\Phi$ is close to zero, then most of the points are much further away than the nearest neighbor, so the latter should be easy to identify. Thus the potential function is an intuitively reasonable measure of the difficulty of NN search.

This general characterization of data configurations amenable to efficient exact NN search, by the three data structures, is our main result. Earlier work has looked at other data structures, and has only provided guarantees for very specific families of data. To illustrate our theorem, we bound $\Phi$ for two commonly-studied data types. In either scenario, the queries are arbitrary.
\begin{itemize}
\item When $x_1, \ldots, x_n$ are drawn i.i.d.\ from a doubling measure (Section~\ref{sec:doubling}). As we discussed earlier, this is the assumption under which many other results for exact NN search have been obtained.
\item When $x_1, \ldots, x_n$ are documents drawn from a topic model (Section~\ref{sec:topics}).
\end{itemize}
For doubling measures of intrinsic dimension $d_o$, we show that the spill tree is able to answer exact nearest neighbor queries in time $O(d_o)^{d_o}  + O(\log n)$, with a probability of error that is an arbitrarily small constant, while the RP tree is slower by only a logarithmic factor (Theorem~\ref{thm:failure-doubling}). These are close to the best results that have been obtained using other data structures. (The failure probability is over the randomization in the tree structure, and can be further reduced by building multiple trees.) We chose the topic model as an example of a significantly harder case: its data distribution is more concentrated, in the sense that there are a lot of data points that are only slightly further away than the nearest neighbor. The resulting savings are far more modest though non-negligible: for large $n$, the time to answer a query is roughly $n \cdot 2^{-O(\sqrt{L})}$, where $L$ is the expected document length.

In some situations, the time to construct the data structure, and the ability to later add or remove data points, are significant factors. It is readily seen that the construction time for the spill tree is proportional to its size, while that of the RP tree and the virtual spill is $O(n \log n)$. Adding and removing points is also easy: all guarantees hold if these are performed locally, while rebuilding the entire data structure after every $O(n)$ such operations.

\section{A potential function for point configurations}

To motivate the potential function $\Phi$, we start by considering what happens when there are just two data points and one query point.

\subsection{How random projection affects the relative placement of three points}

Consider any three points $q, x, y \in \R^d$, such that $x$ is closer to $q$ than is $y$; that is, $\|q-x\| \leq \|q-y\|$.

Now suppose that a random direction $U$ is chosen from the unit sphere $S^{d-1}$, and that the points are projected onto this direction. What is the probability that $y$ falls between $q$ and $x$ on this line? The following lemma answers this question exactly. An approximate solution, with different proof method, was given earlier by \citet{K97}.

\begin{lemma}
Pick any $q,x,y \in \R^d$ with $\|q - x\| \leq \|q - y\|$. Pick a random unit direction $U$. Then
$$ \pr_U(\mbox{$y \cdot U$ falls (strictly) between $q \cdot U$ and $x \cdot U$}) 
\ = \ 
\frac{1}{\pi} \arcsin \left( \frac{\|q-x\|}{\|q-y\|} \sqrt{1 - \left(\frac{(q-x)\cdot(y-x)}{\|q-x\|\,\|y-x\|}\right)^2}\right).$$
\label{lemma:three-points}
\end{lemma}
\begin{proof}
We may assume that $U$ is drawn from $N(0, I_d)$, the $d$-dimensional Gaussian with mean zero and unit covariance. This gives exactly the right distribution if we scale $U$ to unit length, but we can skip this last step since it has no effect on the question we are considering. 

We can also assume, without loss of generality, that $q$ lies at the origin and that $x$ lies along the (positive) $x_1$-axis: that is, $q = 0$ and $x = \|x\| e_1$. It will then be helpful to split the direction $U$ into two pieces, its component $U_1$ in the $x_1$-direction, and the remaining $d-1$ coordinates $U_R$. Likewise, we will write $y = (y_1, y_R)$.

If $y_R = 0$ then $x$, $y$, and $q$ are collinear, and the projection of $y$ cannot possibly fall between those of $x$ and $q$. In what follows, we assume $y_R \neq 0$.

Let $E$ denote the event of interest:
\begin{eqnarray*}
E 
& \equiv & \mbox{$y \cdot U$ falls between $q \cdot U$ (that is, $0$) and $x \cdot U$ (that is, $\|x\| U_1$)} \\
& \equiv & \mbox{$y_R \cdot U_R$ falls between $-y_1 U_1$ and $(\|x\|-y_1) U_1$}
\end{eqnarray*}
The interval of interest is either $(-y_1 |U_1|, (\|x\|-y_1)|U_1|)$, if $U_1 \geq 0$, or $(-(\|x\|-y_1) |U_1|, y_1 |U_1|)$, if $U_1 < 0$. To simplify things, $y_R \cdot U_R$ is independent of $U_1$ and is distributed as $N(0, \|y_R\|^2)$, which is symmetric and thus assigns the same probability mass to the two intervals. We can therefore write
$$ \pr_U (E) \ = \ \pr_{U_1} \pr_{U_R} (-y_1 |U_1| < y_R \cdot U_R < (\|x\| - y_1) |U_1|).$$
Let $Z$ and $Z'$ be independent standard normals $N(0,1)$. Since $U_1$ is distributed as $Z$ and $y_R \cdot U_R$ is distributed as $\|y_R\| Z'$,
$$ \pr_U(E) \ = \ \pr(-y_1 |Z| < \|y_R\| Z' < (\|x\| - y_1) |Z|) \ = \ \pr\left(\frac{Z'}{|Z|} \in \left(-\frac{y_1}{\|y_R\|}, \frac{\|x\|-y_1}{\|y_R\|}\right) \right).$$

Now $Z'/|Z|$ is the ratio of two standard normals, which has a standard Cauchy distribution. Using the formula for a Cauchy density,
\begin{eqnarray*}
\pr(E) 
& = & 
\int_{-y_1/\|y_R\|}^{(\|x\|-y_1)/\|y_R\|} \frac{\mathrm{d}w}{\pi (1+w^2)} \\
& = & 
\frac{1}{\pi} \left( \arctan \left(\frac{\|x\| - y_1}{\|y_R\|}\right) - \arctan \left(\frac{-y_1}{\|y_R\|}\right) \right) \\
& = & 
\frac{1}{\pi} \arctan \frac{\|x\| \, \|y_R\|}{\|y\|^2 - y_1 \|x\|} \\
& = & 
\frac{1}{\pi} \arcsin \left( \frac{\|x\|}{\|y\|} \cdot \sqrt{\frac{\|y\|^2 - y_1^2}{\|y\|^2 + \|x\|^2 - 2y_1 \|x\|}} \right) ,
\end{eqnarray*}
which is exactly the expression in the lemma statement once we invoke $y_1 = (y \cdot x)/\|x\|$ and factor in our assumption that $q=0$.
\end{proof}

To simplify the expression, define an index of the {\it collinearity} of $q,x,y$ to be 
$$ \mbox{\rm coll}(q,x,y) \ = \ \frac{|(q-x) \cdot (y-x)|}{\|q-x\|\,\|y-x\|} .$$
This value, in the range $[0,1]$, is 1 when the points are collinear, and 0 when $q-x$ is orthogonal to $x-y$.
\begin{cor}
Under the conditions of Lemma~\ref{lemma:three-points},
$$ \frac{1}{\pi} \frac{\|q-x\|}{\|q-y\|} 
\sqrt{1 - \mbox{\rm coll}(q,x,y)^2}
\ \leq \ 
\pr_U(\mbox{$y \cdot U$ falls between $q \cdot U$ and $x \cdot U$}) 
\ \leq \ 
\frac{1}{2} \frac{\|q-x\|}{\|q-y\|} .
$$
\label{cor:three-points}
\end{cor}
\begin{proof}
Apply the inequality $\theta \geq \sin \theta \geq 2\theta/\pi$ for all $0 \leq \theta \leq \pi/2$.
\end{proof}

The upper and lower bounds of Corollary~\ref{cor:three-points} are within a constant factor of each other unless the points are approximately collinear.

\subsection{By how much does random projection separate nearest neighbors?}

For a query $q$ and data points $x_1, \ldots, x_n$, let $x_{(1)}, x_{(2)}, \ldots$ denote a re-ordering of the points by increasing distance from $q$. Consider the potential function
$$ \Phi(q, \{x_1, \ldots, x_n\}) \ = \ \frac{1}{n} \sum_{i=2}^n \frac{\|q-x_{(1)}\|}{\|q-x_{(i)}\|} .$$
\begin{thm}
Pick any points $q, x_1, \ldots, x_n \in \R^d$. If these points are projected to a direction $U$ chosen at random from the unit sphere, then
$$ \E_U(\mbox{fraction of the projected $x_i$ that fall between $q$ and $x_{(1)}$})
\ \leq \ \frac{1}{2} \, \Phi(q, \{x_1, \ldots, x_n\}).$$
\label{thm:insertion}
\end{thm}
\begin{proof}
Let $Z_i$ be the event that $x_{(i)}$ falls between $q$ and $x_{(1)}$ in the projection. By Corollary~\ref{cor:three-points}, 
$$ \pr_U(Z_i) 
\ \leq \ 
\frac{1}{2} \frac{\|q - x_{(1)}\|}{\|q - x_{(i)}\|}
$$
The lemma now follows by linearity of expectation.
\end{proof}

The upper bound of Theorem~\ref{thm:insertion} is fairly tight, as can be seen from Corollary~\ref{cor:three-points}, unless there is a high degree of collinearity between the points.

In the tree data structures we analyze, most cells contain only a subset of the data $\{x_1, \ldots, x_n\}$. For a cell that contains $m$ of these points, the appropriate variant of $\Phi$ is
$$ \Phi_m(q, \{x_1, \ldots, x_n\}) \ = \ \frac{1}{m} \sum_{i=2}^m \frac{\|q-x_{(1)}\|}{\|q-x_{(i)}\|} .$$

\begin{cor}
Pick any points $q, x_1, \ldots, x_n$ and let $S$ denote any subset of the $x_i$ that includes $x_{(1)}$. If $q$ and the points in $S$ are projected to a direction $U$ chosen at random from the unit sphere, then for any $0 < \alpha < 1$,
$$ \pr_U(\mbox{at least an $\alpha$ fraction of $S$ falls between $q$ and $x_{(1)}$ when projected})
\ \leq \ 
\frac{1}{2\alpha} \Phi_{|S|}(q, \{x_1, \ldots, x_n\}).$$
\label{cor:insertion}
\end{cor}
\begin{proof}
This follows immediately by applying Theorem~\ref{thm:insertion} to $S$, noting that the corresponding value of $\Phi$ is maximized when $S$ consists of the points closest to $q$, and then applying Markov's inequality.
\end{proof}

\subsection{Extension to $k$ nearest neighbors}

If we are interested in finding the $k$ nearest neighbors, a suitable generalization of $\Phi_m$ is
$$ \Phi_{k,m}(q, \{x_1, \ldots, x_n\}) \ = \ 
\frac{1}{m} \sum_{i=k+1}^m \frac{(\|q - x_{(1)}\| + \cdots + \|q - x_{(k)}\|)/k}{\|q - x_{(i)}\|} .
$$
\begin{thm}
Pick any points $q, x_1, \ldots, x_n$ and let $S$ denote any subset of the $x_i$ that includes $x_{(1)}, \ldots, x_{(k)}$. Suppose $q$ and the points in $S$ are projected to a direction $U$ chosen at random from the unit sphere. Then, for any $0 < \alpha < 1$, the probability (over $U$) that in the projection, there is some $1 \leq j \leq k$ for which $\geq \alpha m$ points lie between $x_{(j)}$ and $q$ is at most
$$
\frac{k}{2(\alpha - (k-1)/|S|)} \Phi_{k,|S|}(q, \{x_1, \ldots, x_n\}).
$$
provided $k < \alpha |S| + 1$.
\label{thm:insertion-k}
\end{thm}
\begin{proof}
Set $m = |S|$. As in Corollary~\ref{cor:insertion}, the probability of the bad event is maximized when $S = \{x_{(1)}, \ldots, x_{(m)}\}$, so we will assume as much.

For any $1 \leq j \leq k$, let $N_j$ denote the number of points in $\{x_{(k+1)}, \ldots, x_{(m)}\}$ that fall (strictly) between $q$ and $x_{(j)}$ in the projection. Reasoning as in Theorem~\ref{thm:insertion}, we have
$$ \pr_U (N_j \geq \alpha m - (k-1)) \ \leq \ \frac{\E_U N_j}{\alpha m - (k-1)} 
\ \leq \ 
\frac{1}{2(\alpha m - (k-1))} \sum_{i=k+1}^m \frac{\|q - x_{(j)}\|}{\|q - x_{(i)}\|} .$$
Taking a union bound over all $1 \leq j \leq k$,
\begin{eqnarray*}
\pr_U (\exists 1 \leq j \leq k: N_j \geq \alpha m - (k-1)) 
& \leq & 
\frac{1}{2(\alpha m - (k-1))} \sum_{i = k+1}^m \frac{\|q - x_{(1)}\| + \cdots + \|q - x_{(k)}\|}{\|q - x_{(i)}\|} \\
& = & 
\frac{k}{2(\alpha - (k-1)/m)} \Phi_{k,m}(q, \{x_1, \ldots, x_n\}),
\end{eqnarray*}
as claimed.
\end{proof}

\subsection{Bounds on $\Phi$}

The results so far suggest that $\Phi$ is closely related to the failure probabilities of the randomized search trees we have described. In the next section, we will make this relationship precise. We will then give bounds on $\Phi$ for various types of data. Here is a brief preview: for large enough $m$, very roughly,
$$
\Phi_m(q, \{x_1, \ldots, x_n\}) \ \leq \ 
\left\{
\begin{array}{ll}
1/m^{1/d_o} & \mbox{doubling measure of intrinsic dimension $d_o$} \\
1/\sqrt{L} & \mbox{topic model with expected document length $L$} 
\end{array}
\right.
$$

\section{Randomized partition trees}

We'll now see that the failure probability of the random projection tree is proportional to $\Phi \ln (1/\Phi)$, while that of the two spill trees is proportional to $\Phi$. We start with the second result, since it is the more straightforward of the two.

\subsection{Randomized spill trees} 

In a randomized spill tree, each cell is split along a direction chosen uniformly at random from the unit sphere. Two kinds of splits are simultaneously considered: (1) a split at the median (along the random direction), and (2) an overlapping split with one part containing the bottom $1/2 + \alpha$ fraction of the cell's points, and the other part containing the top $1/2 + \alpha$ fraction, where $0 < \alpha < 1/2$ (recall Figure~\ref{fig:splits}).

We consider two data structures that use these splits in different ways. The {\it spill tree} stores each data point in (possibly) multiple leaves, using overlapping splits. The tree is grown until each leaf contains at most $n_o$ points. A query is answered by routing it to a single leaf, using median splits, and returning the NN in that leaf.

The time to answer a query is just $O(n_o + \log (n/n_o))$, but the space requirement of this data structure is super-linear. Its depth is $\ell = \log_{1/\beta} n/n_o$ levels, where $\beta = (1/2) + \alpha$, and thus the total size is 
$$n_o 2^\ell 
\ = \ 
n_o \left( \frac{n}{n_o} \right)^{\log_{1/\beta} 2}.
$$
We will take $n_o$ to be a constant independent of $n$, so this size is $O(n^{\log_{1/\beta} 2})$. When $\alpha = 0.05$, for instance, the size is $O(n^{1.159})$. When $\alpha = 0.1$, it is $O(n^{1.357})$.

A {\it virtual spill tree} stores each data point in a single leaf, using median splits, once again growing the tree until each leaf has $n_o$ or fewer points. Thus the total size is just $O(n)$ and the depth is $\log_2 (n/n_o)$. However, a query is answered by routing it to multiple leaves using overlapping splits, and then returning the NN in the union of these leaves.

\begin{thm}
Suppose a randomized spill tree is built using data points $\{x_1, \ldots, x_n\}$, to depth $\ell = \log_{1/\beta} (n/n_o)$, where $\beta = (1/2) + \alpha$ for regular spill trees and $\beta = 1/2$ for virtual spill trees. If this tree is used to answer a query $q$, then the probability (over randomization in the construction of the tree) that it fails to return $x_{(1)}$ is at most
$$ \frac{1}{2\alpha} \, \sum_{i=0}^\ell \Phi_{\beta^i n}(q, \{x_1, \ldots, x_n\}) . $$
The probability that it fails to return the $k > 1$ nearest neighbors $x_{(1)}, \ldots, x_{(k)}$ is at most
$$ \frac{k}{\alpha} \, \sum_{i=0}^\ell \Phi_{k,\beta^i n}(q, \{x_1, \ldots, x_n\}), $$
provided $k \leq \alpha n_o/2$.
\label{thm:spill}
\end{thm}
\begin{proof}
Let's start with the regular spill tree. Consider the internal node at depth $i$ on the root-to-leaf path of query $q$; this node contains $\beta^i n$ data points, for $\beta = (1/2) + \alpha$. What is the probability that $q$ gets separated from $x_{(1)}$ when the node is split? This bad event can only happen if $q$ and $x_{(1)}$ lie on opposite sides of the median and if $x_{(1)}$ is transmitted only to one side of the split, that is, if at least $\alpha$ fraction of the points lie between $x_{(1)}$ and the median. This means that at least an $\alpha$ fraction of the cell's projected points must fall between $q$ and $x_{(1)}$, which occurs with probability at most $(1/2\alpha) \Phi_{\beta^i n}(q, \{x_1, \ldots, x_n\})$ by Corollary~\ref{cor:insertion}. The lemma follows by summing over all levels $i$.

The argument for the virtual spill tree is identical, except that we use $\beta = 1/2$ and we swap the roles of $q$ and $x_{(1)}$; for instance, we consider the root-to-leaf path of $x_{(1)}$.

The generalization to $k$ nearest neighbors is immediate for spill trees. The probability of something going wrong at level $i$ of the tree is, by Theorem~\ref{thm:insertion-k}, at most 
$$ \frac{k}{2(\alpha - (k-1)/n_o)} \Phi_{k, \beta^i n} \ \leq \ \frac{k}{\alpha} \Phi_{k, \beta^i n}.$$
Virtual spill trees require a slightly more careful argument. If the root-to-leaf path of each $x_{(j)}$, for $1 \leq j \leq k$, is considered separately, it can be shown that the total probability of failure at level $i$ is again bounded by the same expression.
\end{proof}

As we mentioned earlier, we will encounter two functional forms of $\Phi_m$: either $1/m^{1/d_o}$, where $d_o$ is a notion of intrinsic dimension, or a small constant $1/\sqrt{L}$. In the former case, the failure probability of the spill tree is roughly $1/(\alpha n_o^{1/d_o})$, and in the latter case it is $(1/(\alpha \sqrt{L})) \log (n/n_o)$. Further details are in Sections~\ref{sec:doubling} and \ref{sec:topics}.

\subsection{Random projection trees}

In an RP tree, a cell is split by choosing a direction uniformly at random from the unit sphere $S^{d-1}$, projecting the points in the cell onto that direction, and then splitting at the $\beta$ fractile, for $\beta$ chosen uniformly at random from $[1/4,3/4]$. As in a $k$-d tree, each point is mapped to a single leaf. Likewise, a query point is routed to a particular leaf, and its nearest neighbor within that leaf is returned. 

In many of the statements below, we will drop the arguments $(q, \{x_1, \ldots, x_n\})$ of $\Phi$ in the interest of readability.
\begin{thm}
Suppose an RP tree is built using points $\{x_1, \ldots, x_n\}$ and is then used to answer a query $q$. The probability (over the randomization in tree construction) that it fails to return the nearest neighbor of $q$ is at most 
$$ \sum_{i=0}^\ell \Phi_{\beta^i n} \ln \frac{2e}{\Phi_{\beta^in}} ,$$
where $\beta = 3/4$ and $\ell = \log_{1/\beta} (n/n_o)$.
The probability that it fails to return the $k$ nearest neighbors of $q$ is at most 
$$ 
\left( 2k \sum_{i=0}^\ell \Phi_{k, \beta^i n} \ln \frac{2e}{k \Phi_{k, \beta^i n}} \right) +  \frac{16(k-1)}{n_o}.
$$
\label{thm:rp}
\end{thm}
\begin{proof}
Consider any internal node of the tree that contains $q$ as well as $m$ of the data points, including $x_{(1)}$. What is the probability that the split at that node separates $q$ from $x_{(1)}$? To analyze this, let $F$ denote the fraction of the $m$ points that fall between $q$ and $x_{(1)}$ along the randomly-chosen split direction. Since the split point is chosen at random from an interval of mass $1/2$, the probability that it separates $q$ from $x_{(1)}$ is at most $F/(1/2)$. Integrating out $F$, we get
\begin{eqnarray*}
\pr(\mbox{$q$ is separated from $x_{(1)}$}) 
& \leq &
\int_0^1 \! \pr(F = f) \frac{f}{1/2} \, \mathrm{d}f \\
& = & 
2 \int_0^1 \! \pr(F > f) \, \mathrm{d}f \\
& \leq & 
2 \int_0^1 \! \min \left( 1, \frac{\Phi_m}{2f} \right) \, \mathrm{d}f \\
& = & 
2 \int_0^{\Phi_m/2} \! \mathrm{d}f + 2 \int_{\Phi_m/2}^1 \! \frac{\Phi_m}{2f} \, \mathrm{d}f 
\ \ = \ \ 
\Phi_m \ln \frac{2e}{\Phi_m},
\end{eqnarray*}
where the second inequality uses Corollary~\ref{cor:insertion}.

The lemma follows by taking a union bound over the path that conveys $q$ from root to leaf, in which the number of data points per level shrinks geometrically, by a factor of $3/4$ or better.

The same reasoning generalizes to $k$ nearest neighbors. This time, $F$ is defined to be the fraction of the $m$ points that lie between $q$ and the furthest of $x_{(1)}, \ldots, x_{(k)}$ along the random splitting direction. Then $q$ is separated from one of these neighbors only if the split point lies in an interval of mass $F$ on either side of $q$, an event that occurs with probability at most $2F/(1/2)$. Using Theorem~\ref{thm:insertion-k},
\begin{eqnarray*}
\lefteqn{
\pr(\mbox{$q$ is separated from some $x_{(j)}$, $1 \leq j \leq k$})} \\
& \leq &
\int_0^1 \! \pr(F = f) \frac{2f}{1/2} \, \mathrm{d}f \\
& = & 
4 \int_0^1 \! \pr(F > f) \, \mathrm{d}f \\
& \leq & 
4 \int_0^1 \! \min \left( 1, \frac{k \Phi_{k,m}}{2(f - (k-1)/m)} \right) \, \mathrm{d}f \\
& \leq & 
4 \int_0^{(k\Phi_{k,m}/2) + (k-1)/m} \! \mathrm{d}f + 4 \int_{(k\Phi_{k,m}/2) + (k-1)/m}^1 \! \frac{k \Phi_{k,m}}{2(f - (k-1)/m)} \, \mathrm{d}f \\
& \leq &
2 k \Phi_{k,m} \ln \frac{2e}{k \Phi_{k,m}} + \frac{4(k-1)}{m},
\end{eqnarray*}
and as before, we sum this over a root-to-leaf path in the tree.
\end{proof}

\subsection{Is randomization necessary?}

The tree data structures we have studied make crucial use of random projection for splitting cells. It would not suffice to use coordinate directions, as in $k$-d trees.

To see this, consider a simple example. Let $q$, the query point, be the origin, and suppose the data points $x_1, \ldots, x_n \in \R^d$ are chosen as follows:
\begin{itemize}
\item $x_1$ is the all-ones vector.
\item Each $x_i, i > 1$, is chosen by picking a coordinate at random, setting its value to $M$, and then setting all remaining coordinates to uniform-random numbers in the range $(0,1)$. Here $M$ is some very large constant.
\end{itemize}
For large enough $M$, the nearest neighbor of $q$ is $x_1$. By letting $M$ grow further, we can let $\Phi(q, \{x_1, \ldots, x_n\})$ get arbitrarily close to zero, which means that our random projection methods will work admirably. However, any coordinate projection will create a disastrously large separation between $q$ and $x_1$: on average, a $(1-1/d)$ fraction of the data points will fall between them.

\section{Bounding $\Phi$}

The exact nearest neighbor schemes we analyze have error probabilities related to $\Phi$, which lies in the range $[0,1]$. The worst case is when all points are equidistant, in which case $\Phi$ is exactly 1, but this is a pathological situation. Is it possible to bound $\Phi$ under simple assumptions on the data?

In this section we study two such assumptions. In each case, query points are arbitrary, but the data are assumed to have been drawn i.i.d.\ from an underlying distribution.

\subsection{Data drawn from a doubling measure}
\label{sec:doubling}

Suppose the data points are drawn from a distribution $\mu$ on $\R^d$ which is a {\it doubling measure}: that is, there exist a constant $C > 0$ and a subset $\X \subseteq \R^d$ such that
$$ \mu(B(x,2r)) \ \leq \ C \cdot \mu(B(x,r)) \ \ \ \ \mbox{for all $x \in \X$ and all $r > 0$}. $$
Here $B(x,r)$ is the closed Euclidean ball of radius $r$ centered at $x$. To understand this condition, it is helpful to also look at an alternative formulation that is essentially equivalent: there exist a constant $d_o > 0$ and a subset $\X \subset \R^d$ such that for all $x \in \X$, all $r > 0$, and all $\alpha \geq 1$,
$$ \mu(B(x,\alpha r)) \ \leq \ \alpha^{d_o} \cdot \mu(B(x,r)) .$$
In other words, the probability mass of a ball grows polynomially in the radius. Comparing this to the standard formula for the volume of a ball, we see that the degree of this polynomial, $d_o$ (which is $\log_2 C$), can reasonably be thought of as the ``dimension'' of the measure $\mu$.

\begin{thm}
Suppose $\mu$ is continuous on $\R^d$ and is a doubling measure with dimension $d_o \geq 2$. Pick any $q \in \X$ and draw $x_1, \ldots, x_n$ independently at random from $\mu$. Pick any $0 < \delta < 1/2$. Then with probability at least $1-3\delta$ over the choice of the $x_i$, for all $2 \leq m \leq n$,
$$ \Phi_m(q, \{x_1, \ldots, x_n\})
\ \leq \ 
6 \left( \frac{2}{m} \ln \frac{1}{\delta} \right)^{1/d_o} .
$$
\label{thm:phi-doubling}
\end{thm}
\begin{proof}
We will consider a collection of balls $B_o, B_1, B_2, \ldots$ centered at $q$, with geometrically increasing radii $r_o, r_1, r_2, \ldots$, respectively. For $i \geq 1$, we will take $r_i = 2^i r_o$. Thus by the doubling condition, $\mu(B_i) \leq C^i \mu(B_o)$, where $C = 2^{d_o} \geq 4$.

Define $r_o$ to be the radius for which $\mu(B(q,r_o)) = (1/n) \ln (1/\delta)$. This choice implies that $x_{(1)}$ is likely to fall in $B_o$: when points $X = \{x_1, \ldots, x_n\}$ are drawn randomly from $\mu$,
$$ \pr(\mbox{no point falls in $B_o$}) 
\ = \ 
(1 - \mu(B_o))^n
\ \leq \ 
\delta
.
$$
Next, for $i \geq 1$, the expected number of points falling in ball $B_i$ is at most $n C^i \mu(B_o) = C^i \ln (1/\delta)$, and by a multiplicative Chernoff bound,
$$ \pr(|X \cap B_i| \geq 2n C^i \mu(B_o))
\ \leq \ 
\exp(- (n C^i \mu(B_o)/3))
\ = \ 
\delta^{C^i/3}
\ \leq \ 
\delta^{iC/3}
.$$
Summing over all $i$, we get
$$
\pr(\exists i \geq 1: |X \cap B_i| \geq 2n C^i \mu(B_o))
\ \leq \ 
2 \delta^{C/3} 
\ \leq \ 
2 \delta
.$$
We will henceforth assume that $x_{(1)}$ lies in $B_o$ and that each $B_i$ has at most $2n\mu(B_o) C^i = 2C^i \ln (1/\delta)$ points.

Pick any $2 \leq m \leq n$, and recall the expression for $\Phi$:
$$ \Phi_m(q, \{x_1, \ldots, x_n\}) \ = \ \frac{1}{m} \sum_{i=2}^m \frac{\|q - x_{(1)}\|}{\|q - x_{(i)}\|} .$$
Once $x_{(1)}$ is fixed, moving other points closer to $q$ can only increase $\Phi$. Therefore, the maximizing configuration has $2n\mu(B_o) C$ points in $B_1$, followed by $2n \mu(B_o) C^2$ points in $B_2$, and then $2n \mu(B_o) C^3$ points in $B_3$, and so on. Each point in $B_j \setminus B_{j-1}$ contributes at most $1/2^{j-1}$ to the $\Phi$ summation.

Under the worst-case configuration, points $x_{(1)}, \ldots, x_{(m)}$ lie within $B_\ell$, for $\ell$ such that
\begin{equation}
2n\mu(B_o)C^{\ell-1} \ < \ m \ \leq \ 2n\mu(B_o) C^\ell.   \tag{*}
\end{equation}
We then have
\begin{eqnarray*}
\Phi_m
& \leq & 
\frac{1}{m} \left(|X \cap B_1| + \left(\sum_{j=2}^{\ell-1} |X \cap (B_j\setminus B_{j-1})| \cdot \frac{1}{2^{j-1}} \right) + (m - |X \cap B_{\ell-1}|) \cdot \frac{1}{2^{\ell-1}}\right) \\
& = & 
\frac{1}{m} \left(|X \cap B_1| + \sum_{j=2}^{\ell-1} \left(\frac{|X \cap B_j|}{2^{j-1}} - \frac{|X \cap B_{j-1}|}{2^{j-1}} \right) + (m - |X \cap B_{\ell-1}|) \cdot \frac{1}{2^{\ell-1}}\right) \\
& = & 
\frac{1}{m} \left(\frac{m}{2^{\ell-1}} + \sum_{j=1}^{\ell-1} \frac{|X \cap B_j|}{2^j} \right) \\
& = & 
\frac{1}{m} \left( \frac{m}{2^{\ell-1}} + 2n\mu(B_o) \sum_{j=1}^{\ell-1} \left(\frac{C}{2}\right)^j \right) \\
& \leq & 
\frac{1}{m} \left( \frac{m}{2^{\ell-1}} + 4n\mu(B_o) \left(\frac{C}{2}\right)^{\ell-1} \right) \\
& \leq & 
\frac{1}{m} \left( \frac{m}{2^{\ell-1}} + \frac{2m}{2^{\ell-1}} \right) 
\ \ = \ \ 
\frac{6}{2^\ell},
\end{eqnarray*}
where the last inequality comes from (*). To lower-bound $2^\ell$, we again use (*) to get $C^\ell \geq m/(2 n \mu(B_o))$, whereupon 
$$ 2^\ell 
\ \geq \ 
\left( \frac{m}{2 n \mu(B_o)} \right)^{1/\log_2 C}
\ = \ 
\left( \frac{m}{2 \ln (1/\delta)} \right)^{1/\log_2 C}
$$
and we're done.
\end{proof}

This extends easily to the potential function for $k$ nearest neighbors.
\begin{thm}
Under the same conditions as Theorem~\ref{thm:phi-doubling}, for any $k \geq 1$, we have
$$ \Phi_{k,m}(q, \{x_1, \ldots, x_n\})
\ \leq \ 
6 \left( \frac{8}{m} \max\left(k , \ln \frac{1}{\delta} \right)\right)^{1/d_o} .
$$
\label{thm:phi-doubling-k}
\end{thm}
\begin{proof}
The only big change is in the definition of $r_o$; it is now the radius for which 
$$ \mu(B_o) \ = \ \frac{4}{n} \max\left(k , \ln \frac{1}{\delta} \right) .$$
Thus, when $x_1, \ldots, x_n$ are drawn independently at random from $\mu$, the
expected number of them that fall in $B_o$ is at least $4k$, and by a multiplicative
Chernoff bound is at least $k$ with probability $\geq 1-\delta$. 

The balls $B_1, B_2, \ldots$ are defined as before, and once again, we can conclude that with
probability $\geq 1-2\delta$, each $B_i$ contains at most $2n C^i \mu(B_o)$ of the data points.

Any point $x_{(i)} \not\in B_o$ lies in some annulus $B_j\setminus B_{j-1}$, and its contribution 
to the summation in $\Phi_{k,m}$ is
$$ \frac{(\|q - x_{(1)}\| + \cdots + \|q - x_{(k)}\|)/k}{\|q - x_{(i)}\|} \ \leq \ \frac{1}{2^{j-1}}.$$
The relationship (*) and the remainder of the argument are exactly as before. 
\end{proof}

We can now give bounds on the failure probabilities of the three tree data structures.
\begin{thm}
There is an absolute constant $c_o$ for which the following holds.
Suppose $\mu$ is a doubling measure on $\R^d$ of intrinsic dimension $d_o \geq 2$.
Pick any query $q \in \X$ and draw $x_1, \ldots, x_n$ independently from $\mu$. Then
with probability at least $1-3\delta$ over the choice of data:
\begin{enumerate}
\item[(a)] For either variant of the spill tree, if $k \leq \alpha n_o/2$,
$$
\pr(\mbox{spill tree fails to return $k$ nearest neighbors}) 
\ \  \leq \ \  
\frac{c_o d_o k}{\alpha} \left( \frac{8 \max(k, \ln 1/\delta)}{n_o} \right)^{1/d_o} .
$$
\item[(b)] For the RP tree with $n_o \geq c_o (3k)^{d_o} \max(k, \ln 1/\delta)$, 
$$
\pr(\mbox{RP tree fails to return $k$ nearest neighbors})
\ \ \leq \ \  
c_o k (d_o + \ln n_o) \left(\frac{8 \max(k, \ln 1/\delta)}{n_o} \right)^{1/d_o}.
$$
\end{enumerate}
These probabilities are over the randomness in tree construction.
\label{thm:failure-doubling}
\end{thm}
\begin{proof}
These bounds follow immediately from Theorems~\ref{thm:spill}, \ref{thm:rp}, 
and \ref{thm:phi-doubling-k}, using Lemma~\ref{lemma:bounding-summation} from 
the appendix to bound the summation.
\end{proof}

In order to make the failure probability an arbitrarily small constant, it is sufficient
to take $n_o = O(d_o k)^{d_o} \max(k, \ln 1/\delta)$ for spill trees and 
$n_o = O(d_o k \ln (d_o k))^{d_o} \max(k, \ln 1/\delta)$ for RP trees.

\subsection{A document model}
\label{sec:topics}

In a bag-of-words model, a document is represented as a binary vector in $\{0,1\}^N$, where $N$ is the size of the vocabulary and the $i$th coordinate is 1 if the document happens to contain the corresponding word. This is a sparse representation in which the number of nonzero positions is typically much smaller than $N$.

Pick any query document $q \in \{0,1\}^N$, and suppose that $x_1, \ldots, x_n$ are generated i.i.d.\ from a topic model $\mu$. We will consider a simple such model with $t$ topics, each of which follows a product distribution. The distribution $\mu$ is parametrized by the mixing weights over topics, $w_1, \ldots, w_t$, which sum to one, and the word probabilities $(p_1^{(j)}, \ldots, p_N^{(j)})$ for each topic $1 \leq j \leq t$. Here is the generative process for a document $x$:
\begin{itemize}
\item Pick a topic $1 \leq j \leq t$, where the probability of picking $j$ is $w_j$.
\item Set the coordinates of $x \in \{0,1\}^N$ independently; the $i$th coordinate is 1 with probability $p_i^{(j)}$.
\end{itemize}
The overall distribution is thus a mixture $\mu = w_1 \mu_1 + \cdots + w_t \mu_t$ whose $j$th component is a Bernoulli product distribution $\mu_j = B(p_1^{(j)}) \times \cdots \times B(p_N^{(j)})$. Here $B(p)$ is a shorthand for the distribution on $\{0,1\}$ with expected value $p$. It will simplify things to assume that $0 < p_i^{(j)} < 1/2$; this is not a huge assumption if, say, stopwords have been removed.

For the purposes of bounding $\Phi$, we are interested in the distribution of $d_H(q, X)$, where $X$ is chosen from $\mu$ and $d_H$ denotes Hamming distance. This is a sum of small independent quantities, and it is customary to approximate such sums by a Poisson distribution. In the current context, however, this approximation is rather poor, and we instead use counting arguments to directly bound how rapidly the distribution grows. The results stand in stark contrast to those we obtained for doubling measures, and reveal this to be a substantially more difficult setting for nearest neighbor search. For a doubling measure, the probability mass of a ball $B(q, r)$ doubles whenever $r$ is multiplied by a constant. In our present setting, it doubles whenever $r$ is increased by an additive constant. Specifically, it turns out (Lemma~\ref{lem:fast-growth}) that for $\ell \leq L/8$,
$$ \frac{\pr(d_H(q, X) = \ell+1)}{\pr(d_H(q,X) = \ell)} \ \geq \ 4 .$$
Here $L = \min(L_1, \ldots, L_t)$, where $L_j$ is the expected number of words in a document drawn from $\mu_j$, that is, $L_j = p_1^{(j)} + \cdots + p_N^{(j)}$.

We start with the case of a single topic.

\subsubsection{Growth rate for one topic}

Let $q \in \{0,1\}^N$ be any fixed document and let $X$ be drawn from a Bernoulli product distribution $B(p_1) \times \cdots \times B(p_N)$. Then the Hamming distance $d_H(q,X)$ is distributed as a sum of Bernoullis,
$$ d_H(q, X) \sim B(a_1) + \cdots + B(a_N) ,$$
where
$$ 
a_i 
\ = \ 
\left\{
\begin{array}{ll}
p_i   & \mbox{if $q_i = 0$} \\
1-p_i & \mbox{if $q_i = 1$}
\end{array}
\right.
$$

To understand this distribution, we start with a general result about sums of Bernoulli random variables. Notice that the result is exactly correct in the situation where all $p_i = 1/2$.
\begin{lemma}
Suppose $Z_1, \ldots, Z_N$ are independent, where $Z_i \in \{0,1\}$ is a Bernoulli random variable with mean $0 < a_i < 1$, and $a_1 \geq a_2 \geq \cdots \geq a_N$. Let $Z = Z_1 + \cdots + Z_N$. Then for any $\ell \geq 0$,
$$ \frac{\pr(Z = \ell+1)}{\pr(Z = \ell)} \ \ \geq \ \ \frac{1}{\ell+1} \sum_{i=\ell+1}^N \frac{a_i}{1-a_i} .$$
\label{lem:bernoulli-sum}
\end{lemma}

\begin{proof}
Define $r_i = a_i/(1-a_i) \in (0, \infty)$; then $r_1 \geq r_2 \geq \cdots \geq r_N$. Now, for any $\ell \geq 0$,
\begin{eqnarray*}
\pr(Z = \ell) 
& = & 
\sum_{\mbox{\scriptsize $\{i_1, \ldots, i_\ell\} \subset [N]$}} a_{i_1} a_{i_2} \cdots a_{i_\ell} \prod_{j \not\in \{i_1, \ldots, i_\ell\}} (1-a_j) \\
& = & 
\prod_{i=1}^N (1-a_i) \sum_{\mbox{\scriptsize $\{i_1, \ldots, i_\ell\}\subset [N]$}} \frac{a_{i_1}}{1-a_{i_1}} \frac{a_{i_2}}{1-a_{i_2}} \cdots \frac{a_{i_\ell}}{1-a_{i_\ell}} \\
& = & 
\prod_{i=1}^N (1-a_i) \sum_{\mbox{\scriptsize $\{i_1, \ldots, i_\ell\}\subset [N]$}} r_{i_1} r_{i_2} \cdots r_{i_\ell}
\end{eqnarray*}
where the summations are over subsets $\{i_1, \ldots, i_\ell\}$ of $\ell$ distinct elements of $[N]$. In the final line, the product of the $(1-a_i)$ does not depend upon $\ell$ and can be ignored. Let's focus on the summation; call it $S_\ell$. We would like to compare it to $S_{\ell + 1}$.

$S_{\ell + 1}$ is the sum of ${N \choose \ell+1}$ distinct terms, each the product of $\ell+1$ $r_i$'s. These terms also appear in the quantity $S_\ell (r_1 + \cdots + r_N)$; in fact, each term of $S_{\ell+1}$ appears multiple times, $\ell+1$ times to be precise. The remaining terms in $S_\ell(r_1 + \cdots + r_N)$ each contain $\ell-1$ unique elements and one duplicated element. By accounting in this way, we get
\begin{eqnarray*}
S_\ell (r_1 + \cdots + r_N) 
& = & 
(\ell+1) S_{\ell+1} + \sum_{\mbox{\scriptsize $\{i_1, \ldots, i_\ell\} \subset [N]$}} r_{i_1} r_{i_2} \cdots r_{i_\ell} (r_{i_1} + \cdots + r_{i_\ell}) \\
& \leq & 
(\ell + 1) S_{\ell+1} + S_\ell (r_1 + \cdots + r_\ell)
\end{eqnarray*}
since the $r_i$'s are arranged in decreasing order. Hence
$$ 
\frac{\pr(Z = \ell+1)}{\pr(Z = \ell)}
\ = \ 
\frac{S_{\ell+1}}{S_\ell}
\ \geq \ 
\frac{1}{\ell+1} (r_{\ell+1} + \cdots + r_N),
$$
as claimed.
\end{proof}

We now apply this result directly to the sum of Bernoulli variables $Z = d_H(q,X)$.
\begin{lemma}
Suppose that $p_1, \ldots, p_N \in (0, 1/2)$. Pick any query $q \in \{0,1\}^N$, and draw $X$ from distribution $\mu = B(p_1) \times \cdots \times B(p_N)$. Then for any $\ell \geq 0$,
$$ \frac{\pr(d_H(q, X) = \ell+1)}{\pr(d_H(q,X) = \ell)} \ \ \geq \ \ \frac{L - \ell/2}{\ell+1},$$
where $L = \sum_i p_i$ is the expected number of words in $X$.
\label{lem:fast-growth}
\end{lemma}

\begin{proof}
Suppose $q$ contains $k_o$ nonzero entries. Without loss of generality, these are $q_1, \ldots, q_{k_o}$.

As we have seen, $d_H(q,X)$ is distributed as the Bernoulli sum $B(1-p_1) + \cdots + B(1-p_{k_o}) + 
B(p_{k_o+1}) + \cdots + B(p_N)$. Define
$$ r_i \ = \ 
\left\{
\begin{array}{ll}
(1-p_i)/p_i & \mbox{if $i \leq k_o$} \\
p_i/(1-p_i) & \mbox{if $i > k_o$}
\end{array}
\right.
$$
Notice that $r_i > 1$ for $i \leq k_o$, and $\leq 1$ for $i > k_o$; and that $r_i > p_i$ always.

By Lemma~\ref{lem:bernoulli-sum}, we have that for any $\ell \geq 0$,
$$
\frac{\pr(d_H(q, X) = \ell+1)}{\pr(d_H(q,X) = \ell)} \ \geq \ \frac{1}{\ell+1} \sum_{i > \ell} r_{(i)},
$$
where $r_{(1)} \geq \cdots \geq r_{(N)}$ denotes the reordering of $r_1, \ldots, r_N$ into
descending order. Since each $r_i > p_i$, and each $p_i$ is at most $1/2$,
$$ \sum_{i > \ell} r_{(i)} 
\ \geq \ 
(\mbox{sum of $N-\ell$ smallest $p_i$'s}) 
\ \geq \ 
(\sum_i p_i) - \ell/2 
\ = \ 
L - \ell/2.
$$
\end{proof}

\subsubsection{Growth rate for multiple topics}

Now let's return to the original model, in which $X$ is chosen from a mixture of $t$ topics $\mu = w_1 \mu_1 + \cdots + w_t \mu_t$, with $\mu_j = B(p_1^{(j)}) \times \cdots \times B(p_N^{(j)})$. Then for any $\ell$,
$$ \pr(d_H(q, X) = \ell\ |\ X \sim \mu) 
\ = \ 
\sum_{j=1}^t w_j \pr(d_H(q,X) = \ell\ |\  X \sim \mu_j) .
$$
Combining this relation with Lemma~\ref{lem:fast-growth}, we immediately get the following.
\begin{cor}
Suppose that all $p_i^{(j)} \in (0,1/2)$. Let $L_j = \sum_i p_i^{(j)}$ denote the expected number of words in a document from topic $j$, and let $L = \min(L_1, \ldots, L_t)$. Pick any query $q \in \{0,1\}^N$, and draw $X \sim \mu$. For any $\ell \geq 0$,
$$ \frac{\pr(d_H(q, X) = \ell+1)}{\pr(d_H(q,X) = \ell)} \ \ \geq \ \ \frac{L - \ell/2}{\ell+1}.$$
\label{cor:fast-growth}
\end{cor}

\subsubsection{Bounding $\Phi$}

Fix a particular query $q \in \{0,1\}^N$, and draw $x_1, \ldots, x_n$ from distribution $\mu$. Let the random variable $S_\ell$ denote the points at Hamming distance exactly $\ell$ from $q$, so that $\E |S_\ell| = n \pr_{X \sim \mu}(d_H(q,X) = \ell)$. 

\begin{lemma}
There is an absolute constant $c_o$ for which the following holds.
Pick any $0 < \delta < 1$ and any $k \geq 1$, and let $v$ denote the smallest integer for which
$\pr_{X \sim \mu}(d_H(q,X) \leq v) \geq (8/n) \max(k, \ln 1/\delta)$. 
Then with probability at least $1-3\delta$,
\begin{enumerate}
\item[(a)] $|S_0| + \cdots + |S_v| \geq k$.
\item[(b)] If $v \leq c_o L$ then $|S_0| + \cdots + |S_{v-1}| \leq |S_v|$.
\item[(c)] For all $v \leq \ell \leq c_oL$, we have $|S_{\ell + 1}|/|S_\ell| \geq 2$.
\end{enumerate}
If (a, b, c) hold, then for any $m \leq n$,
$$ \Phi_{k,m}(q, \{x_1, \ldots, x_n\})
\ \leq \ 
4 \sqrt{\frac{v}{c_o L - \log_2 (n/m)}}.$$
\label{lemma:empirical-growth}
\end{lemma}

\begin{proof}
Parts (a, b, c) are shown by applying multiplicative Chernoff bounds to the result of Corollary~\ref{cor:fast-growth}. The details are very similar to those of Theorem~\ref{thm:phi-doubling-k}, and hence we omit them and turn to bounding $\Phi$.

Suppose that for some $i > k$, point $x_{(i)}$ is at Hamming distance $\ell$ from $q$, that is, $x_{(i)} \in S_{\ell}$. Then
$$ 
\frac{(\|q-x_{(1)}\| + \cdots + \|q - x_{(k)}\|)/k}{\|q - x_{(i)}\|} 
\ \leq \ 
\sqrt{\frac{v}{\ell}}
$$
since Euclidean distance is the square root of Hamming distance. In bounding $\Phi_{k,m}$, we need to gauge the range of Hamming distances spanned by $x_{(k+1)}, \ldots, x_{(m)}$.

The geometric growth rate of part (c) implies that most points lie at Hamming distance $c_o L$ or greater from $q$. It also means that $d_H(q, x_{(m)}) > c_o L - \log_2 (n/m)$. Thus,
\begin{eqnarray*}
\Phi_{k,m}(q, \{x_1, \ldots, x_n\}) 
& = &
\frac{1}{m} \sum_{i>k} \frac{(\|q-x_{(1)}\| + \cdots + \|q - x_{(k)}\|)/k}{\|q - x_{(i)}\|} \\
& \leq &
\frac{1}{m} \sum_{\ell \geq v} | S_\ell \cap \{x_{(1)}, \ldots, x_{(m)}\}| \sqrt{\frac{v}{\ell}}
\ \ \leq \ \  
4 \sqrt{\frac{v}{c_oL - \log_2 (n/m)}}
\end{eqnarray*}
where the last step follows by lower-bounding $|S_\ell|$ by an increasing geometric series.
\end{proof}

The implication of this lemma is that for any of the three tree data structures, the failure probability at a single level is roughly $\sqrt{v/L}$. This means that the tree can only be grown to depth $O(\sqrt{L/v})$, and thus the query time is dominated by $n_o = n \cdot 2^{-O(\sqrt{L/v})}$.

When $n$ is large, we expect $v$ to be small, and thus the query time improves over exhaustive search by a factor of roughly $2^{-\sqrt{L}}$.

\bibliographystyle{apalike}
\bibliography{/Users/dasgupta314/linux/PAPERS/sanjoy}

\begin{thebibliography}{}

\bibitem[Andoni and Indyk, 2008]{AI08}
Andoni, A. and Indyk, P. (2008).
\newblock Near-optimal hashing algorithms for approximate nearest neighbor in
  high dimensions.
\newblock {\em Communications of the ACM}, 51(1):117--122.

\bibitem[Arya et~al., 1998]{AMNSW98}
Arya, S., Mount, D., Netanyahu, N., Silverman, R., and Wu, A. (1998).
\newblock An optimal algorithm for approximate nearest neighbor searching.
\newblock {\em Journal of the ACM}, 45:891--923.

\bibitem[Bentley, 1975]{B75}
Bentley, J. (1975).
\newblock Multidimensional binary search trees used for associative searching.
\newblock {\em Communications of the ACM}, 18(9):509--517.

\bibitem[Beygelzimer et~al., 2006]{BKL06}
Beygelzimer, A., Kakade, S., and Langford, J. (2006).
\newblock Cover trees for nearest neighbor.
\newblock In {\em 23rd International Conference on Machine Learning}.

\bibitem[Dasgupta and Freund, 2008]{DF08}
Dasgupta, S. and Freund, Y. (2008).
\newblock Random projection trees and low dimensional manifolds.
\newblock In {\em ACM Symposium on Theory of Computing}, pages 537--546.

\bibitem[Karger and Ruhl, 2002]{KR02}
Karger, D. and Ruhl, M. (2002).
\newblock Finding nearest neighbors in growth-restricted metrics.
\newblock In {\em ACM Symposium on Theory of Computing}, pages 741--750.

\bibitem[Kleinberg, 1997]{K97}
Kleinberg, J. (1997).
\newblock Two algorithms for nearest-neighbor search in high dimensions.
\newblock In {\em 29th ACM Symposium on Theory of Computing}.

\bibitem[Krauthgamer and Lee, 2004]{KL04}
Krauthgamer, R. and Lee, J. (2004).
\newblock Navigating nets: simple algorithms for proximity search.
\newblock In {\em ACM-SIAM Symposium on Discrete Algorithms}.

\bibitem[Liu et~al., 2004]{LMGY04}
Liu, T., Moore, A., Gray, A., and Yang, K. (2004).
\newblock An investigation of practical approximate nearest neighbor
  algorithms.
\newblock In {\em Neural Information Processing Systems}.

\bibitem[Maneewongvatana and Mount, 2001]{MM01}
Maneewongvatana, S. and Mount, D. (2001).
\newblock The analysis of a probabilistic approach to nearest neighbor
  searching.
\newblock In {\em Seventh International Worshop on Algorithms and Data
  Structures}, pages 276--286.

\end{thebibliography}

\appendix

\section{Technical lemma}

\begin{lemma}
Suppose that for some constants $A, B > 0$ and $d_o \geq 1$,
$$ F(m) \ \leq \ A \left( \frac{B}{m} \right)^{1/d_o} $$
for all $m \geq n_o$. Pick any $0 < \beta < 1$ and define $\ell = \log_{1/\beta} (n/n_o)$. Then:
$$
\sum_{i = 0}^\ell F(\beta^i n)
\ \  \leq \ \  \frac{A d_o}{1-\beta} \left( \frac{B}{n_o} \right)^{1/d_o} 
$$
and, if $n_o \geq B(A/2)^{d_o}$,
$$
\sum_{i=0}^\ell F(\beta^i n) \ln \frac{2e}{F(\beta^i n)}
\ \ \leq \ \ 
\frac{A d_o}{1-\beta} \left( \frac{B}{n_o} \right)^{1/d_o} 
\left( \frac{1}{1-\beta} \ln \frac{1}{\beta} + \ln \frac{2e}{A} + \frac{1}{d_o} \ln \frac{n_o}{B} \right) 
.
$$
\label{lemma:bounding-summation}
\end{lemma}
\begin{proof}
Writing the first series in reverse,
\begin{eqnarray*}
\sum_{i = 0}^{\ell} F(\beta^i n)
\ = \ 
\sum_{i=0}^\ell F\left( \frac{n_o}{\beta^i} \right)  
& \leq & 
\sum_{i=0}^\ell A \left( \frac{B \beta^i}{n_o} \right)^{1/d_o} \\
& = & 
A \left( \frac{B}{n_o} \right)^{1/d_o} \sum_{i=0}^\ell \beta^{i/d_o} \\
& \leq & 
\frac{A}{1-\beta^{1/d_o}} \left( \frac{B}{n_o} \right)^{1/d_o} 
\ \leq \ 
\frac{A d_o}{1-\beta}  \left( \frac{B}{n_o} \right)^{1/d_o}.
\end{eqnarray*}
The last inequality is obtained by using
$$ (1-x)^p \geq 1-px \mbox{\ \ \ for $0 < x < 1$, $p \geq 1$} $$
to get $(1 - (1-\beta)/d_o)^{d_o} \geq \beta$ and thus $1-\beta^{1/d_o} \geq (1-\beta)/d_o$.

Now we move on to the second bound. The lower bound on $n_o$ implies that $A (B/m)^{1/d_o} \leq 2$
for all $m \geq n_o$. Since $x \ln (2e/x)$ is increasing when $x \leq 2$, we have
\begin{eqnarray*}
\sum_{i=0}^\ell F(\beta^i n) \ln \frac{2e}{F(\beta^i n)}
& \leq & 
\sum_{i=0}^\ell A \left( \frac{B}{\beta^i n} \right)^{1/d_o} \ln \frac{2e}{A (B/(\beta^i n))^{1/d_o}} .
\end{eqnarray*}
The lemma now follows from algebraic manipulations that invoke the first bound as well as the 
inequality
$$ \sum_{i = 0}^\ell i F \left( \frac{n_o}{\beta^i} \right)
\ \leq \ 
\frac{A d_o^2}{(1-\beta)^2} \left( \frac{B}{n_o} \right)^{1/d_o} ,
$$
which in turn follows from
$$
\sum_{i = 0}^\ell i \beta^{i/d_o} 
\ \leq \ 
\sum_{i=1}^\infty i \beta^{i/d_o}
\ = \ 
\sum_{i=1}^\infty \sum_{j = i}^\infty \beta^{j/d_o} 
\ = \ 
\sum_{i=1}^\infty \frac{\beta^{i/d_o}}{1-\beta^{1/d_o}}
\ = \ 
\frac{\beta^{1/d_o}}{(1-\beta^{1/d_o})^2}
\ \leq \ 
\frac{d_o^2}{(1-\beta)^2}.
$$
\end{proof}

\end{document}